\documentclass[a4paper, 12pt, subeqn]{article}

\usepackage[english, russian]{babel}

\usepackage{amsmath}
\usepackage{amssymb}
\usepackage{amsthm}

\usepackage{graphicx,color}

\newcommand{\sgn}{\mathrm{sgn}}
\newcommand{\const}{\mathrm{const}}

\newcommand{\Tr}{\mathrm{Tr}}

\renewcommand{\Im}{\mathrm{Im}}
\renewcommand{\Re}{\mathrm{Re}}

\newtheorem{lemma}{Lemma}
\newtheorem{theorem}{Theorem}
\newtheorem{proposition}{Proposition}

\newtheorem{statement}{Statement}

\numberwithin{equation}{section}
\numberwithin{lemma}{section}
\numberwithin{theorem}{section}
\numberwithin{proposition}{section}
\numberwithin{statement}{section}

\begin{document}

\selectlanguage{english}
\begin{center}
\large Absence of exponentially localized solitons for the Novikov--Veselov equation at negative energy
\end{center}

\begin{center}
A.V. Kazeykina \footnote{Centre des Math\'ematiques Appliqu\'ees, Ecole Polytechnique, Palaiseau, 91128, France; \\ email: kazeykina@cmap.polytechnique.fr}, R.G. Novikov\footnote{CNRS (UMR 7641),  Centre des Math\'ematiques Appliqu\'ees, Ecole Polytechnique, Palaiseau, 91128, France; email: novikov@cmap.polytechnique.fr}
\end{center}

\textbf{Abstract.} We show that Novikov--Veselov equation (an analog of KdV in dimension $ 2 + 1 $) does not have exponentially localized solitons at negative energy.

\section{Introduction}
In the present paper we are concerned with the following $ ( 2 + 1 ) $--dimensional analog of the Korteweg--de Vries equation:
\begin{equation}
\label{NV}
\begin{aligned}
& \partial_t v = 4 \Re ( 4 \partial_z^3 v + \partial_z( v w ) - E \partial_z w ), \\
& \partial_{ \bar z } w = - 3 \partial_{z} w, \quad v = \bar v, \quad E \in \mathbb{R}, \\
& v = v( x, t ), \quad w = w( x, t ), \quad x = ( x_1, x_2 ) \in \mathbb{R}^2, \quad t \in \mathbb{R},
\end{aligned}
\end{equation}
where
\begin{equation}
\label{derivatives}
\partial_t = \frac{ \partial }{ \partial t }, \quad \partial_z = \frac{ 1 }{ 2 } \left( \frac{ \partial }{ \partial x_1 } - i \frac{ \partial }{ \partial x_2 } \right), \quad \partial_{ \bar z } = \frac{ 1 }{ 2 } \left( \frac{ \partial }{ \partial x_1 } + i \frac{ \partial }{ \partial x_2 } \right).
\end{equation}

Equation (\ref{NV}) is contained implicitly in the paper of S.V. Manakov \cite{M} as an equation possessing the following representation
\begin{equation}
\frac{ \partial( L - E ) }{ \partial t } = [ L - E, A ] + B( L - E )
\end{equation}
(Manakov L--A--B triple), where $ L = - \Delta + v( x, t ) $, $ \Delta = 4 \partial_{z} \partial_{ \bar z } $, $ A $ and $ B $ are suitable differential operators of the third and zero order respectively, $ [ \cdot, \cdot ] $ denotes the commutator. Equation (\ref{NV}) was written in an explicit form by S.P. Novikov and A.P. Veselov in \cite{NV1}, \cite{NV2}, where higher analogs of (\ref{NV}) were also constructed. Note that both Kadomtsev--Petviashvili equations can be obtained from (\ref{NV}) by considering an appropriate limit $ E \to \pm \infty $ (see \cite{ZS}, \cite{G}).

In the case when $ v( x_1, x_2, t ) $, $ w( x_1, x_2, t ) $ are independent of $ x_2 $, (\ref{NV}) can be reduced to the classic KdV equation:
\begin{equation}
\label{KdV}
\partial_t u - 6 u \partial_x u + \partial_x^3 u = 0, \quad x \in \mathbb{R}, \quad t \in \mathbb{R}.
\end{equation}
It is well--known that (\ref{KdV}) has the soliton solutions
\begin{equation}
\label{ch}
u( x, t ) = u_{ \kappa, \varphi }( x - 4 \kappa^2 t ) = - \frac{ 2 \kappa^2 }{ ch^2( \kappa( x - 4 \kappa^2 t - \varphi ) ) }, \quad x \in \mathbb{R}, \, t \in \mathbb{R}, \, \kappa \in ( 0, +\infty ), \, \varphi \in \mathbb{R}.
\end{equation}
Evidently,
\begin{equation}
\label{1dim_sol}
\begin{aligned}
& u_{ \kappa, \varphi } \in C^{ \infty }( \mathbb{R} ), \\
& \partial_{ x }^j u_{ \kappa, \varphi }( x ) = O( e^{ -2 \kappa | x | } ) \text{ as } x \to \infty, \quad j = 0, 1, 2, \ldots
\end{aligned}
\end{equation}
Properties (\ref{1dim_sol}) imply that the solitons (\ref{ch}) are exponentially localized in $ x $.

For the $ 2 $--dimensional case we will say that a solution $ ( v, w ) $ of (\ref{NV}) is an exponentially localized soliton if the following properties hold:
\begin{equation}
\label{soliton}
\begin{aligned}
& v( x, t ) = V( x - ct ), \quad x \in \mathbb{R}^2, \quad c = ( c_1, c_2 ) \in \mathbb{R}^2, \\
& V \in C^3( \mathbb{R}^2 ), \quad \partial_x^j V( x ) = O( e^{ - \alpha | x | } ) \text{ for } | x | \to \infty, \quad | j | \leqslant 3 \text{ and some  } \alpha > 0 \\
& (\text{where } j = ( j_1, j_2 ) \in ( 0 \cup \mathbb{N} )^2, | j | = | j_1 | + | j_2 |, \quad \partial^j_x = \partial^{ j_1 + j_2 } / \partial x_1^{ j_1 } \partial x_2^{ j_2 } ), \\
& w( \cdot, t ) \in C( \mathbb{R}^2 ), \quad w( x, t ) \to 0 \text{ as } | x | \to \infty, \quad t \in \mathbb{R}.
\end{aligned}
\end{equation}

In \cite{N1} it was shown that, in contrast with the $ ( 1 + 1 ) $--dimensional case, the $ ( 2 + 1 ) $--dimensional KdV equation (\ref{NV}), at least for $ E = E_{ fixed } > 0 $, does not have exponentially localized solitons. More precisely, in \cite{N1} it was shown that the following theorem is valid for $ E = E_{ fixed } > 0 $:
\begin{theorem}
\label{main_theorem}
Let $ (v, w) $ be an exponentially localized soliton solution of (\ref{NV}) in the sense (\ref{soliton}). Then $ v \equiv 0 $, $ w \equiv 0 $.
\end{theorem}

The main result of this paper consists in the proof of Theorem \ref{main_theorem} for the case $ E = E_{ fixed } < 0 $. This proof is given in Section \ref{theorem_section} and is based on Propositions \ref{sol_b_zero} and \ref{b_zero_v_zero}. In addition: Proposition \ref{sol_b_zero} is an analog of the result of \cite{N1} about the transparency of sufficiently localized solitons for equation (\ref{NV}) for $ E > 0 $; Proposition \ref{b_zero_v_zero} is an analog of the result of \cite{N2}, \cite{GN} that there are no nonzero bounded real--valued exponentially localized transparent potentials (that is potentials with zero scattering amplitude) for the Schr\"odinger equation (\ref{schrodinger}) for $ E = E_{ fixed } > 0 $.

Note that nonzero bounded algebraically localized solitons for equation (\ref{NV}) for $ E < 0 $ are also unknown (see [G]), but their absence is not proved.

As regards integrable systems in $ 2 + 1 $ dimensions admitting exponentially decaying solitons in all directions on the plane, see \cite{BLMP}, \cite{FS}.

As regards integrable systems in  $ 2 + 1 $ dimensions admitting nonzero bounded algebraically decaying solitons in all directions on the plane, see \cite{FA}, \cite{BLMP}, \cite{G}, \cite{KN} and references therein.

\section{Inverse scattering for the 2--dimensional Schr\"odinger equation at a fixed negative energy}
Consider the scattering problem for the two--dimensional Schr\"odinger equation at a fixed negative energy:
\begin{equation}
\label{schrodinger}
\begin{aligned}
& - \Delta \psi + v( z ) \psi = E \psi, \quad E = E_{ fixed } < 0, \\
& \Delta = 4 \partial_{ z } \partial_{ \bar z }, \quad z = x_1 + i x_2, \quad x \in \mathbb{R}^2,
\end{aligned}
\end{equation}
where $ \partial_z $, $ \partial_{ \bar z } $ are the same as in (\ref{derivatives}).
We will assume that the potential $ v( z ) $ satisfies the following conditions
\begin{equation}
\label{v_conditions}
\begin{aligned}
& v( z ) = \overline{ v( z ) }, \quad v( z ) \in L^{ \infty }( \mathbb{C} ), \\
& | v( z ) | < q ( 1 + | z | )^{ - 2 - \varepsilon } \text{ for some } q > 0, \; \varepsilon > 0.
\end{aligned}
\end{equation}
In this paper we will be concerned with the exponentially decreasing potentials, i.e. with the potentials $ v( z ) $ satisfying (\ref{v_conditions}) and the following additional condition
\begin{equation}
\label{exp_decrease}
v( z ) = O( e^{ - \alpha | z | } ) \text{ as } | z | \to \infty \text{ for some } \alpha > 0.
\end{equation}

Direct and inverse scattering for the two--dimensional Schr\"odinger equation (\ref{schrodinger}) at fixed negative energy under assumptions (\ref{v_conditions}) was considered for the first time in \cite{GspN}. For some of the results discussed in this section see also \cite{N2}, \cite{G}.

First of all, we note that by scaling transform we can reduce the scattering problem with an arbitrary fixed negative energy to the case when $ E = -1 $. Therefore, in our further reasoning we will assume that $ E = -1 $.

It is known that for $ \lambda \in \mathbb{C} \backslash ( 0 \cup \mathcal{E} ) $, where
\begin{equation}
\label{e_set}
\begin{aligned}
& \mathcal{E} \text{ is the set of zeros of the modified Fredholm determinant } \Delta \\
& \text{ for the integral equation (\ref{fr_equation}),}
\end{aligned}
\end{equation}
there exists a unique continuous solution $ \psi( z, \lambda ) $ of (\ref{schrodinger}) with the following asymptotics
\begin{equation}
\label{psi_mu}
\psi( z, \lambda ) = e^{ -\frac{ 1 }{ 2 }( \lambda \bar z + z / \lambda ) } \mu( z, \lambda ), \quad \mu( z, \lambda ) = 1 + o( 1 ), \quad | z | \to \infty.
\end{equation}

In addition, the function $ \mu( z, \lambda ) $ satisfies the following integral equation
\begin{align}
\label{mu_int_equation}
& \mu( z, \lambda ) = 1 + \iint\limits_{ \zeta \in \mathbb{C} } g( z - \zeta, \lambda ) v( \zeta ) \mu( \zeta, \lambda ) d \zeta_R d \zeta_I, \\
\label{green}
& g( z, \lambda ) = - \left( \frac{ 1 }{ 2 \pi } \right)^2 \iint\limits_{ \zeta \in \mathbb{C} } \frac{ \exp( i / 2 ( \zeta \bar z + \bar \zeta z ) ) }{ \zeta \bar \zeta + i ( \lambda \bar \zeta + \zeta / \lambda ) } d \zeta_R d \zeta_I,
\end{align}
where $ z \in \mathbb{C} $, $ \lambda \in \mathbb{C} \backslash 0 $, $ \zeta_R = \Re \zeta  $, $ \zeta_I = \Im \zeta $.

In terms of $ \psi $ of (\ref{psi_mu}) equation (\ref{mu_int_equation}) takes the form
\begin{align}
\label{psi_int_equation}
& \psi( z, \lambda ) = e^{ -1/2 ( \lambda \bar z + z / \lambda ) } + \iint\limits_{ \zeta \in \mathbb{C} } G( z - \zeta, \lambda ) v( \zeta ) \psi( \zeta, \lambda ) d \zeta_R d \zeta_I, \\
\label{g_big}
& G( z, \lambda ) = e^{ -1/2 ( \lambda \bar z + z / \lambda ) } g( z, \lambda ),
\end{align}
where $ z \in \mathbb{C} $, $ \lambda \in \mathbb{C} \backslash 0 $.

In terms of  $ m( z, \lambda ) = ( 1 + | z | )^{ -( 2 + \varepsilon ) / 2 } \mu( z, \lambda ) $ equation (\ref{mu_int_equation}) takes the form
\begin{equation}
\label{fr_equation}
m( z, \lambda ) = ( 1 + | z | )^{ -( 2 + \varepsilon ) / 2 } + \iint\limits_{ \zeta \in \mathbb{C} } ( 1 + | z | )^{ -( 2 + \varepsilon ) / 2 } g( z - \zeta, \lambda ) \frac{ v( \zeta ) }{ ( 1 + | \zeta | )^{ -( 2 + \varepsilon ) / 2 } } m( \zeta, \lambda ) d \zeta_R d \zeta_I,
\end{equation}
where $ z \in \mathbb{C} $, $ \lambda \in \mathbb{C} \backslash 0 $. In addition, $ A( \cdot, \cdot, \lambda ) \in L^2( \mathbb{C} \times \mathbb{C} ) $, $ | \Tr A^2( \lambda ) | < \infty $, where $ A( z, \zeta, \lambda ) $ is the Schwartz kernel of the integral operator $ A( \lambda ) $ of the integral equation (\ref{fr_equation}). Thus, the modified Fredholm determinant for (\ref{fr_equation}) can be defined by means of the formula:
\begin{equation}
\label{fr_det}
\ln \Delta( \lambda ) = \Tr( \ln( I - A( \lambda ) ) + A( \lambda ) )
\end{equation}
(see \cite{GK} for more precise sense of such definition).



Taking the subsequent members in the asymptotic expansion (\ref{psi_mu}) for $ \psi( z, \lambda ) $, we obtain (see \cite{N2}):
\begin{multline}
\label{ab_def}
\psi( z, \lambda ) = \exp\left( - \frac{ 1 }{ 2 } \left( \lambda \bar z + \frac{ z }{ \lambda } \right) \right) \Biggl\{ 1 - 2 \pi \sgn ( 1 - \lambda \bar \lambda ) \times \\
\times \left( \frac{ i \lambda a( \lambda ) }{ z - \lambda^2 \bar z } + \exp\left( - \frac{ 1 }{ 2 } \left( \left( \frac{ 1 }{ \bar \lambda } - \lambda \right) \bar z + \left( \frac{ 1 }{ \lambda } - \bar \lambda \right) z \right) \right) \frac{ \bar \lambda b( \lambda ) }{ i ( \bar \lambda^2 z - \bar z ) } \right) + o\left( \frac{ 1 }{ | z | } \right)\Biggl\},
\end{multline}
$ | z | \to \infty $, $ \lambda \in \mathbb{C} \backslash ( \mathcal{E} \cup 0 ) $.

The functions $ a( \lambda ) $, $ b( \lambda ) $ from (\ref{ab_def}) are called the "scattering" data for the problem (\ref{schrodinger}), (\ref{v_conditions}) with $ E = - 1 $. It is known that for $ a( \lambda ) $, $ b( \lambda ) $ the following formulas hold (see \cite{N2}):
\begin{align}
\label{a_formula}
& a( \lambda ) = \left( \frac{ 1 }{ 2 \pi } \right)^2 \iint\limits_{ z \in \mathbb{C} } \mu( z, \lambda ) v( z ) d z_R d z_I, \\
\label{b_formula}
& b( \lambda ) = \left( \frac{ 1 }{ 2 \pi } \right)^2 \iint\limits_{ z \in \mathbb{C} } \exp\left( - \frac{ 1 }{ 2 } \left( \left( \lambda - \frac{ 1 }{ \bar \lambda } \right) \bar z - \left( \bar \lambda - \frac{ 1 }{ \lambda } \right) z \right) \right) \mu( z, \lambda ) v( z ) d z_R d z_I,
\end{align}
where $ \lambda \in \mathbb{C} \backslash ( 0 \cup \mathcal{E} ) $, $ z_R = \Re z $, $ z_I = \Im z $. In addition, formally, formulas (\ref{a_formula}), (\ref{b_formula}) can be written as
\begin{equation}
\label{abh}
a( \lambda ) = h( \lambda, \lambda ), \quad b( \lambda ) = h\left( \lambda, \frac{ 1 }{ \bar \lambda } \right),
\end{equation}
where
\begin{equation}
\label{h}
h( \lambda, \lambda' ) = \left( \frac{ 1 }{ 2 \pi } \right)^2 \iint\limits_{ z \in \mathbb{C} } \exp\left( \frac{ 1 }{ 2 } \left( \lambda' \bar z + z / \lambda' \right) \right) \psi( z, \lambda ) v( z ) d z_R z_I,
\end{equation}
and $ \lambda \in \mathbb{C} \backslash ( 0 \cup \mathcal{E} ) $, $ \lambda' \in \mathbb{C} \backslash 0 $. (Note that, under assumptions (\ref{v_conditions}), the integral in (\ref{h}) is well--defined if $ \lambda' = \lambda $ of if $ \lambda' = 1 / \bar \lambda $ but is not well--defined in general.)

Let
\begin{equation}
\label{t_set}
T = \{ \lambda \in \mathbb{C} \colon | \lambda | = 1 \}.
\end{equation}
From (\ref{abh}), in particular, the following statement follows:
\begin{statement}
\label{abT}
Let (\ref{v_conditions}) hold and $ \Delta \neq 0 $ on $ T $. Then
\begin{equation}
\label{ab_equality}
a( \lambda ) = b( \lambda ), \quad \lambda \in T.
\end{equation}
\end{statement}

The following properties of functions $ \Delta( \lambda ) $, $ a( \lambda ) $, $ b( \lambda ) $ will play a substantial role in the proof of Theorem \ref{main_theorem}.
\begin{statement}
\label{delta_prop}
Let (\ref{v_conditions}) hold. Then:
\begin{enumerate}
\item \label{continuity} $ \Delta( \lambda ) \in C( \mathbb{C} ) $;
\item \label{limits} $ \Delta( \lambda ) \to 1 $ as $ \lambda \to 0 $ or $ \lambda \to \infty $;
\item \label{constT} $ \Delta( \lambda ) \equiv \const $ for $ \lambda \in T $;
\item \label{real_valued} $ \Delta $ is real--valued: $ \Delta = \bar \Delta $.
\item \label{delta_sym} $ \Delta( \lambda ) = \Delta( 1 / \bar \lambda ) $, $ \lambda \in \mathbb{C} \backslash 0 $.
\end{enumerate}
\end{statement}

\begin{statement}
\label{real_analytic}
Let conditions (\ref{v_conditions})--(\ref{exp_decrease}) be fulfilled. Then:
\begin{itemize}

\item $ \Delta( \lambda ) $ is a real--analytic function on $ D_+ $, $ D_- $, where
\begin{equation}
\label{d_domains}
D_+ = \{ \lambda \in \mathbb{C} \colon 0 < | \lambda | \leqslant 1 \}, \quad D_- = \{ \lambda \in \mathbb{C} \colon | \lambda | \geqslant 1 \}.
\end{equation}

\item $ a( \lambda ) = \dfrac{ \mathcal{A}( \lambda ) }{ \Delta( \lambda ) } $, $ b( \lambda ) = \dfrac{ \mathcal{B}( \lambda ) }{ \Delta( \lambda ) } $, where $ \mathcal{A}( \lambda ) $, $ \mathcal{B}( \lambda ) $ are real--analytic functions on $ D_+ $, $ D_- $.

\end{itemize}
\end{statement}

Items \ref{continuity}--\ref{real_valued} of Statement \ref{delta_prop} are either known or follow from results mentioned in \cite{HN}, \cite{N2} (see page 129 of \cite{HN} and pages 420, 423, 429 of \cite{N2}). In particular, item 1 of Statement \ref{delta_prop} is a consequence of continuous dependency of $ g( z, \lambda ) $ on $ \lambda \in \mathbb{C} \backslash 0 $; item 3 of Statement \ref{delta_prop} is a consequence of (\ref{fr_det}) and of the formula (see pages 420, 423 of \cite{N2}) $ G( z, \lambda ) = ( -i / 4 ) H_0^1( i | z | ) $, $ z \in \mathbb{C} $, $ \lambda \in T $, where $ G $ is defined by (\ref{g_big}), $ H_0^1 $ is the Hankel function of the first type. In addition, item \ref{delta_sym} of Statement \ref{delta_prop} follows from item \ref{real_valued} of this statement and from symmetry $ \overline{G( z, \lambda )} = G( z, 1 / \bar \lambda ) $, $ z \in \mathbb{C} $, $ \lambda \in \mathbb{C} \backslash 0 $.

Statement \ref{real_analytic} is similar to Proposition 4.2 of \cite{N2} and follow from: (i) formulas (\ref{a_formula}), (\ref{b_formula}), (ii) Cramer type formulas for solving the integral equation (\ref{fr_equation}), (iii) the analog of Proposition 3.2 of \cite{N2} for $ g $ of (\ref{green}).

Under assumptions (\ref{v_conditions}), the function $ \mu( z, \lambda ) $, defined by (\ref{mu_int_equation}), satisfies the following properties:
\begin{equation}
\label{mu_continuity}
\mu( z, \lambda ) \text{ is a continuous function of } \lambda \text{ on } \mathbb{C} \backslash ( 0 \cup \mathcal{E} );
\end{equation}

\begin{subequations}
\label{mu_props}
\begin{align}
\label{mu_dbar}
& \frac{ \partial \mu( z, \lambda ) }{ \partial \bar \lambda } = r( z, \lambda ) \overline{ \mu( z, \lambda ) }, \\
\label{r}
& r( z, \lambda ) = r( \lambda ) \exp\left( \frac{ 1 }{ 2 } \left( \left( \lambda - \frac{ 1 }{ \bar \lambda } \right) \bar z - \left( \bar \lambda - \frac{ 1 }{ \lambda } \right) z \right) \right), \\
\label{t}
& r( \lambda ) = \frac{ \pi \sgn( 1 - \lambda \bar \lambda ) }{ \bar \lambda } b( \lambda )
\end{align}
\end{subequations}
for $ \lambda \in \mathbb{C} \backslash ( 0 \cup \mathcal{E} ) $;

\begin{equation}
\label{mu_limits}
\mu \to 1, \text{ as } \lambda \to \infty, \; \lambda \to 0.
\end{equation}

The function $ b $ possesses the following properties (see \cite{GspN}, \cite{N2}):
\begin{gather}
\label{b_cont}
b \in C( \mathbb{C} \backslash \mathcal{E} ), \\
\label{b_sym}
b\left( - \frac{ 1 }{ \overline{\lambda} } \right) = b( \lambda ), \quad b\left( \frac{ 1 }{ \overline{\lambda} } \right) = \overline{ b( \lambda ) }, \quad \lambda \in \mathbb{C} \backslash 0, \\
\label{b_decay}
\lambda^{-1} b( \lambda ) \in L_p( D_+ ) \text{ (as a function of $ \lambda $) if } \mathcal{E} = \varnothing, \quad 2 < p < 4.
\end{gather}

In addition, the following theorem is valid:
\begin{theorem}[\cite{GspN}, \cite{N2}]
Let $ v $ satisfy (\ref{v_conditions}) and $ \mathcal{E} = \varnothing $ for this potential. Then $ v $ is uniquely determined by its scattering data $ b $ (by means of (\ref{mu_continuity}), (\ref{mu_props}) and equation (\ref{NV}) for $ E = -1 $ and $ \psi $ of (\ref{psi_mu}) ).
\end{theorem}

Finally, if $ ( v( z, t ), w( z, t ) ) $ is a solution of equation (\ref{NV}) with $ E = -1 $, where $ ( v( z, t ), w( z, t ) ) $ satisfy the following conditions:
\begin{align}
\label{regularity_decay}
& v, w \in C( \mathbb{R}^2 \times \mathbb{R} ) \text{ and for each $ t \in \mathbb{R} $ the following properties are fulfilled:}\nonumber \\
& v( \cdot, t ) \in C^3( \mathbb{R}^2 ), \quad \partial_{ x }^{ j } v( x, t ) = O\left( | x |^{ - 2 - \varepsilon } \right) \text{ for } | x | \to \infty, | j | \leq 3 \text{ and some } \varepsilon > 0, \\
& w( x, t ) \to 0 \text{ for } | x | \to \infty, \nonumber
\end{align}
then the dynamics of the scattering data is described by the following equations
\begin{align}
\label{t_dynamics_a}
& a( \lambda, t ) = a( \lambda, 0 ), \\
\label{t_dynamics_b}
& b( \lambda, t ) = \exp\left\{ \left( \lambda^3 + \frac{ 1 }{ \lambda^3 } - \bar \lambda^3 - \frac{ 1 }{ \bar \lambda^3 } \right) t \right\} b( \lambda, 0 ),
\end{align}
where $ \lambda \in \mathbb{C} \backslash 0 $, $ t \in \mathbb{R} $.

\section{Proof of Theorem \ref{main_theorem}}
\label{theorem_section}
The proof of Theorem \ref{main_theorem} is based on Proposition \ref{sol_b_zero} and Proposition \ref{b_zero_v_zero} given below.

\begin{lemma}
\label{shift_lemma}
Let $ v( z ) $ satisfy (\ref{v_conditions}) and $ a( \lambda ) $, $ b( \lambda ) $ be the scattering data corresponding to $ v( z ) $. Then the scattering data $ a_{ \zeta }( \lambda ) $, $ b_{ \zeta }( \lambda ) $ for the potential $ v_{ \zeta }( z ) = v( z - \zeta ) $ are related to $ a( \lambda ) $, $ b( \lambda ) $ by the formulas
\begin{align}
\label{a_shift}
& a_{ \zeta }( \lambda ) = a( \lambda ), \\
\label{b_shift}
& b_{ \zeta }( \lambda ) = \exp\left( -\frac{ 1 }{ 2 } \left( \left( \lambda - \frac{ 1 }{ \bar \lambda } \right) \bar \zeta - \left( \bar \lambda - \frac{ 1 }{ \lambda } \right) \zeta \right) \right) b( \lambda ),
\end{align}
where $ z, \zeta \in \mathbb{C} $, $ \lambda \in \mathbb{C} \backslash 0 $.
\end{lemma}
\begin{proof}
We first note that $ \psi( z - \zeta, \lambda ) $ satisfies equation (\ref{schrodinger}) with the operator $ L = -\Delta + v_{ \zeta }( z ) $. Then the function $ \psi_{ \zeta }( z, \lambda ) $ corresponding to $ v_{ \zeta }( z ) $ and possessing the asymptotics (\ref{psi_mu}) is $ \psi_{ \zeta }( z, \lambda ) = e^{ -\frac{ 1 }{ 2 }\left( \lambda \bar \zeta + \zeta / \lambda \right) } \psi( z - \zeta, \lambda ) $. In terms of function $ \mu $ this relation is written $ \mu_{ \zeta }( z, \lambda ) = \mu( z - \zeta, \lambda ) $ Thus we have
\begin{equation*}
a_{ \zeta }( \lambda ) = \left( \frac{ 1 }{ 2 \pi } \right)^2 \iint\limits_{ z \in \mathbb{C} } v( z - \zeta ) \mu( z - \zeta, \lambda ) d z_R d z_I = a( \lambda ),
\end{equation*}
and, similarly,
\begin{multline*}
b_{ \zeta }( \lambda ) = \left( \frac{ 1 }{ 2 \pi } \right)^2 \iint\limits_{ z \in \mathbb{C} } \exp\left( -\frac{ 1 }{ 2 } \left( \left( \lambda - \frac{ 1 }{ \bar \lambda } \right) \bar z - \left( \bar \lambda - \frac{ 1 }{ \lambda } \right) z \right) \right) \times \\
\times v( z - \zeta ) \mu( z - \zeta, \lambda ) d z_R d z_I = \\
= \exp\left( -\frac{ 1 }{ 2 } \left( \left( \lambda - \frac{ 1 }{ \bar \lambda } \right) \bar \zeta - \left( \bar \lambda - \frac{ 1 }{ \lambda } \right) \zeta \right) \right) b( \lambda ).
\end{multline*}
\end{proof}

\begin{proposition}
\label{sol_b_zero}
Let $ ( v( z, t ), w( z, t ) ) $ be an exponentially localized soliton of (\ref{NV}) in the sense (\ref{soliton}). Let $ b( \lambda, t ) $ be the scattering data for $ v( z, t ) $ for some $ E = E_{ fixed } < 0 $. Then $ b( \lambda, t ) \equiv 0 $ in the domain where it is well--defined, i.e. in $ \mathbb{C} \backslash \mathcal{E} $, where $ \mathcal{E} $ is defined by (\ref{e_set}).
\end{proposition}
\begin{proof}
In virtue of (\ref{t_dynamics_b}) and Statement \ref{real_analytic} it is sufficient to prove that $ b( \lambda, 0 ) \equiv 0 $ in some neighborhoods of $ 0 $ and $ \infty $.

Let $ U_{ 0 } $, $ U_{ \infty } $ be the neighborhoods of $ 0 $ and  $ \infty $, respectively, such that $ \Delta \neq 0 $ in $ U_{ 0 } $, $ U_{ \infty } $ (such neighborhoods exist  in virtue of item \ref{limits} of Statement \ref{delta_prop}).
For $ \lambda \in U_{ 0 } \cup U_{ \infty } $ the function $ b( \lambda, 0 ) $ is well--defined and continuous. As $ ( v( z, t ), w( z, t ) ) $ is a soliton, the dynamics of the function $ b( \lambda, t ) $ can be written as
\begin{equation}
b( \lambda, t ) = \exp\left( - \frac{ 1 }{ 2 }\left( \left( \lambda - \frac{ 1 }{ \bar \lambda } \right) \bar c - \left( \bar \lambda - \frac{ 1 }{ \lambda } \right) c \right) t \right) b( \lambda, 0 )
\end{equation}
(see Lemma \ref{shift_lemma}).

Combining this with formula (\ref{t_dynamics_b}), we obtain
\begin{multline*}
\exp\left\{ - \frac{ 1 }{ 2 }\left( \left( \lambda - \frac{ 1 }{ \bar \lambda } \right) \bar c - \left( \bar \lambda - \frac{ 1 }{ \lambda } \right) c \right) t \right\} b( \lambda, 0 ) = \\
= \exp\left\{ \left( \lambda^3 + \frac{ 1 }{ \lambda^3 } - \bar \lambda^3 - \frac{ 1 }{ \bar \lambda^3 } \right) t \right\} b( \lambda, 0 ).
\end{multline*}
As functions $ \lambda $, $ \bar \lambda $, $ \frac{ 1 }{ \lambda } $, $ \frac{ 1 }{ \overline{ \lambda } } $, $ \lambda^3 $, $ \bar \lambda^3 $, $ \frac{ 1 }{ \lambda^3 } $, $ \frac{ 1 }{ \overline{\lambda}^3 } $, $ 1 $ are linearly independent in any neighborhood of $ 0 $ and $ \infty $, we obtain that $ b( \lambda, 0 ) \equiv 0 $ in $ U_{ 0 } \cup U_{ \infty } $.
\end{proof}

\begin{proposition}
\label{b_zero_v_zero}
Let $ v( z ) $ satisfy (\ref{v_conditions})--(\ref{exp_decrease}) and $ b( \lambda ) $ be its scattering data for some $ E = E_{ fixed } < 0 $. If $ b( \lambda ) \equiv 0 $ in the domain where it is well--defined, i.e. in $ \mathbb{C} \backslash \mathcal{E} $, where $ \mathcal{E} $ is defined by (\ref{e_set}), then $ v \equiv 0 $.
\end{proposition}

Note that Proposition \ref{b_zero_v_zero} can be considered as an analog of Corollary 3 of \cite{GN}.

\begin{proof}[Proof of Proposition \ref{b_zero_v_zero}]
\begin{enumerate}

\item  First we will prove that from the assumptions of this proposition it follows that $ a( \lambda ) \equiv 0 $ in $ U_0 \cup U_{ \infty } $, where $ U_0 $ and $ U_{ \infty } $ are such neighborhoods of $ 0 $ and $ \infty $, respectively, that $ \Delta( \lambda ) \neq 0 $ for $ \lambda \in U_0 \cup U_{ \infty } $. We note that from (\ref{a_formula}), (\ref{b_formula}), (\ref{mu_dbar}), (\ref{mu_limits}) it follows that
\begin{gather}
\label{a_dbar}
\frac{ \partial a ( \lambda ) }{ \partial \bar \lambda } = \frac{ \pi \sgn( 1 - \lambda \bar \lambda ) }{ \bar \lambda } b( \lambda ) \overline{ b( \lambda ) }, \quad \lambda \in \mathbb{C} \backslash ( \mathcal{E} \cup 0 ), \\
\label{a_limits}
a( \lambda ) \to \hat v( 0 ) \text{ as } \lambda \to \infty \text{ or } \lambda \to 0, \text{ where} \\
\label{v_fourier}
\hat v( p ) = \left( \frac{ 1 }{ 2 \pi } \right)^2 \iint\limits_{ z \in \mathbb{C} } e^{ \frac{ i }{ 2 } ( \bar p z + p \bar z ) } v( z ) d z_R d z_I, \quad p \in \mathbb{C}.
\end{gather}
It means that
\begin{equation}
\label{a_holomorphic}
a( \lambda ) \text{ is holomorphic in } \mathbb{C} \backslash ( \mathcal{E} ).
\end{equation}

According to item \ref{constT} of Statement \ref{delta_prop}, $ \Delta( \lambda ) \equiv \const $ for $ \lambda \in T $. We will consider separately two cases: $ \Delta \equiv C \neq 0 $ on $ T $ and $ \Delta \equiv 0 $ on $ T $.

\begin{itemize}
\item[(a)] $ \Delta( \lambda ) \equiv C \neq 0 $ on $ T $:

From item \ref{continuity} of Statement \ref{delta_prop} it follows that there exists $ U_T $, a neighborhood of $ T $, such that $ \Delta( \lambda ) \neq 0 $ in $ U_T $. Thus $ a( \lambda ) $ is holomorphic in $ U_T $. From Statement \ref{abT} we obtain that $ a( \lambda ) = b( \lambda ) = 0 $ on $ T $. It follows then that $ a( \lambda ) \equiv 0 $ in $ U_T $. Using statement \ref{real_analytic}, we obtain that $ a( \lambda ) \equiv 0 $ in $ U_0 \cup U_{ \infty } $.

\item[(b)] $ \Delta( \lambda ) \equiv 0 $ on $ T $:

In \cite{HN} the $ \bar \partial $--equation for $ \Delta $ was derived. In variables $ \lambda $, $ \bar \lambda $ it is written as
\begin{equation}
\label{dbar_delta}
\frac{ \partial \ln \Delta ( \lambda ) }{ \partial \bar \lambda } = - \frac{ \pi \sgn ( \lambda \bar \lambda - 1 ) }{ \bar \lambda } \left( a\left( \frac{ 1 }{ \bar \lambda } \right) - \hat v( 0 ) \right).
\end{equation}

Equation (\ref{dbar_delta}) and properties (\ref{a_limits}), (\ref{a_holomorphic}) imply that $ \frac{ \partial \ln \Delta }{ \partial \bar \lambda } $ is an antiholomorphic function in $ U_{ 0 } \cup U_{ \infty } $, where $ \Delta $ is close to $ 1 $ and, thus, $ \ln \Delta $ is a well--defined one--valued function. As $ \Delta $ is a real--valued real analytic function, it follows that
\begin{equation}
\ln \Delta = f( \lambda ) + \overline{ f( \lambda ) }
\end{equation}
for some holomorphic function $ f( \lambda ) $, or
\begin{equation}
\label{Frepresentation}
\Delta = F( \lambda ) \overline{F( \lambda )}
\end{equation}
for some holomorphic function $ F( \lambda ) $ on $ U_{ 0 } \cup U_{ \infty } $. Now we will use the following lemma (the proof of this lemma is given in Section \ref{proof_section}):
\end{itemize}
\end{enumerate}

\begin{lemma}
\label{auxiliary_lemma}
Let $ \Delta( \lambda ) $ be real--analytic in $ D_+ = \{ \lambda \in \mathbb{C} \colon 0 < | \lambda | \leq 1 \} $. Suppose that $ \Delta( \lambda ) $ can be represented as
\begin{equation}
\label{mod_representation}
\Delta( \lambda ) = F( \lambda ) \overline{ F( \lambda ) }, \quad \lambda \in U_0,
\end{equation}
for some function $ F( \lambda ) $ holomorphic on $ U_0 $, a neighborhood of zero. Then the representation (\ref{mod_representation}) holds on $ D_+ $, i.e. $ F( \lambda ) $ can be extended analytically to $ D_+ $.
\end{lemma}

\begin{itemize}
\item[]
\begin{itemize}
\item[]
Thus, the representation (\ref{Frepresentation}) is valid separately on $ D_+ $ and on $ D_- $, where we used also item \ref{delta_sym} of Statement \ref{delta_prop}. As $ \Delta( \lambda ) \equiv 0 $ on $ T $, it follows that $ F( \lambda ) \equiv 0 $ on $ T $ and, further, $ F( \lambda ) \equiv 0 $ on $ \mathbb{C} $. This contradicts with item \ref{limits} of Statement \ref{delta_prop}. Thus we have shown that under the assumptions of Proposition \ref{b_zero_v_zero} the case $ \Delta( \lambda ) \equiv 0 $ on $ T $ cannot hold.

\end{itemize}

\item[2.] Our next step is to prove that $ \Delta( \lambda ) \equiv 1 $ for $ \lambda \in \mathbb{C} $.

Formula (\ref{a_limits}) states that $ a( 0 ) = a( \infty ) = \hat v( 0 ) $. Thus from equation (\ref{dbar_delta}) it follows that $ \frac{ \partial \ln \Delta }{ \partial \overline{ \lambda } } = 0 $, and $ \ln \Delta $ is holomorphic in some neighborhood of $ 0 $ and $ \infty $. As $ \Delta( \lambda ) $ is a real--valued function and item \ref{limits} of Statement \ref{delta_prop} holds, we conclude that $ \Delta \equiv 1 $ in some neighborhood of $ 0 $ and $ \infty $. Now using Statement \ref{real_analytic}, we obtain that $ \Delta \equiv 1 $ in $ \mathbb{C} $ and, as a corollary, $ \mathcal{E} = \varnothing $.

\item[3.] From the previous item it follows that equation (\ref{mu_dbar}) holds for $ \forall \lambda $: $ \lambda \in \mathbb{C} \backslash 0 $. Due to the assumptions of Proposition \ref{b_zero_v_zero} and the property that $ \mathcal{E} = \varnothing $, we have that $ b \equiv 0 $ on $ \mathbb{C} $ which means that $ \mu( z, \lambda ) $ is holomorphic on $ D_+ $, $ D_- $. As it is also continuous on $ \mathbb{C} $ and property (\ref{mu_limits}) holds, we conclude that $ \mu( z, \lambda ) \equiv 1 $ and $ v( z ) \equiv 0 $.
\end{itemize}
\end{proof}

\begin{proof}[Proof of Theorem \ref{main_theorem} for $ E = E_{fixed} < 0 $] The result follows immediately from Propositions \ref{sol_b_zero}, \ref{b_zero_v_zero}.
\end{proof}
\section{Proof of Lemma \ref{auxiliary_lemma}}
\label{proof_section}

As $ F( \lambda ) $ is analytic in $ U_0 $, it can be represented in this domain by a Taylor series. Let us consider the radius of convergence $ R $ of this Taylor series. Suppose that the statement of Lemma \ref{auxiliary_lemma} is not true and $ R < 1 $.

Let us take a point $ \lambda_0 $, such that $ | \lambda_0 | = R $. In this point $ \Delta( \lambda ) $ can be represented by the following series
\begin{equation}
\label{delta_series}
\Delta( \lambda ) = \sum_{ k, j = 0 }^{ \infty } b_{ k, j } ( \lambda - \lambda_0 )^k ( \bar \lambda - \bar \lambda_0 )^j
\end{equation}
uniformly convergent in $ U_{ \lambda_0 } $, some neighborhood of $ \lambda_0 $. We will prove that the coefficients $ b_{ j, k } $ satisfy the following properties:
\begin{align*}
& b_{ k, k } \in \mathbb{R}, \quad b_{ k, k } \geq 0;& &(a) \\
& b_{ k, j } = \overline{ b_{ j, k } };& &(b) \\
& b_{ k, j } b_{ m, l } = b_{ k, l } b_{ m, j }.& &(c)
\end{align*}

Indeed,
\begin{align*}
& (a) \colon b_{ k, k } = \left. \frac{ 1 }{ ( k! )^2 } \partial_{ \lambda }^k \partial_{ \bar \lambda }^k \Delta( \lambda ) \right|_{ \lambda = \lambda_0 } = \lim_{ \lambda \to \lambda_0 } \frac{ 1 }{ ( k! )^2 } \partial_{ \lambda }^k \partial_{ \bar \lambda }^k \Delta( \lambda ) = \lim_{ \lambda \to \lambda_0 } \frac{ 1 }{ ( k! )^2 } | \partial_{ \lambda }^k F( \lambda ) |^2 \in \mathbb{R}, \geq 0. \\
& (b) \colon b_{ k, j } = \lim_{ \lambda \to \lambda_0 } \frac{ 1 }{ k! j! } \partial_{ \lambda }^k F( \lambda ) \overline{ \partial_{ \lambda }^j F( \lambda ) } = \lim_{ \lambda \to \lambda_0 } \overline{ \frac{ 1 }{ k! j! } \overline{ \partial_{ \lambda }^k F( \lambda ) } \partial_{ \lambda }^j F( \lambda ) } = \overline{ b_{ j, k } }. \\
& (c) \colon b_{ k, j } b_{ m, l } = \lim_{ \lambda \to \lambda_0 } \frac{ 1 }{ k! j! m! l! } \partial_{ \lambda }^k F( \lambda ) \overline{ \partial_{ \lambda }^j F( \lambda ) } \partial_{ \lambda }^m F( \lambda ) \overline{ \partial_{ \lambda }^l F( \lambda ) } = b_{ k, l } b_{ m, j }.
\end{align*}

From properties (a)--(c) it follows that there exist such $ a_k \in \mathbb{C} $, $ k = 0, 1, \ldots $, that
\begin{equation}
\label{decomposition}
b_{ k, j } = a_k \bar a_j.
\end{equation}

We will prove this statement by considering two different cases:
\begin{enumerate}
\item $ b_{ k, k } = 0 $ $ \forall k \in \mathbb{N} \cup 0 $

In this case from properties (b), (c) it follows that $ b_{ k, j } = 0 $ $ \forall k, j \in \mathbb{N} \cup 0 $, and we can take $ a_k = 0 $ $ \forall k \in \mathbb{N} \cup 0 $.

\item $ b_{ k, k } \neq 0 $ for some $ k \in \mathbb{N} \cup 0 $.

In this case we take $ l $ to be the minimal number such that $ b_{ l, l } \neq 0 $. Then we set $ a_0 = a_1 = \ldots = a_{ l - 1 } = 0 $ and we take an arbitrary complex number $ a_l $ satisfying $ | a_l |^2 = b_{ l, l } $. For the rest of the coefficients we set
\begin{equation}
\label{build_a_n}
a_{ n } = \dfrac{ b_{ n, l } }{ \bar a_{ l } },
\end{equation}
where $ n = l + 1, l + 2, \ldots $.

Now let us prove property (\ref{decomposition}). Let us suppose that $ k < l $. Then $ a_k = 0 $, $ b_{ k, k } = 0 $ and from properties (b), (c) it follows that $ b_{ k, j } = 0 $ $ \forall j \in \mathbb{N} \cup 0 $. Thus property (\ref{decomposition}) holds when $ k < l $. A similar reasoning can be carried out when $ j < l $. Now let us suppose that $ k \geq l $, $ j \geq l $. Then
\begin{equation}
a_k \bar a_j = \frac{ b_{ k, l } \bar b_{ j, l } }{ \bar a_{ l } a_{ l } } = \frac{ b_{ k, l } b_{ l, j } }{ b_{ l, l } } = b_{ k, j }.
\end{equation}
Representation (\ref{decomposition}) is proved.

\end{enumerate}

From convergence of series (\ref{delta_series}) it follows that the following series
\begin{equation}
\label{f_series}
F_1( \lambda ) = \sum_{ k = 0 }^{ \infty } a_k ( \lambda - \lambda_0 )^k
\end{equation}
converges uniformly in $ U_{ \lambda_0 } $ (indeed, the case when $ b_{ k, k } = 0 $ $ \forall k \in \mathbb{N} \cup 0 $ is trivial, and in the case when $ \exists l \colon b_{ l, l } \neq 0 $ we take the sum of the members of series (\ref{delta_series}) with coefficients $ b_{ k, l } $, $ k = 0, 1, \ldots $, and obtain series (\ref{f_series}) multiplied by $ \bar a_l ( \bar \lambda - \bar \lambda_0 )^l $). Thus there exists the function $ F_1( \lambda ) $ analytic in $ U_0 $ such that $ \Delta( \lambda ) = F_1( \lambda ) \overline{ F_1( \lambda ) } $. Consequently, we have two functions $ F( \lambda ) $ and $ F_1( \lambda ) $ analytic in a common domain lying in $ \{ \lambda \in \mathbb{C} \colon | \lambda | \leq R \} \cap U_{ \lambda_0 } $ and such that $ | F( \lambda ) | = | F_1( \lambda ) | $. It means that $ F( \lambda ) $ and $ F_1( \lambda ) $ are equal up to a constant factor: $ F( \lambda ) = \mu F_1( \lambda ) $, $ | \mu | = 1 $. It follows then that $ \mu F_1( \lambda ) $ is an analytic continuation of $ F( \lambda ) $ to $ U_{ \lambda_0 } $.

The same reasoning can be applied to any point $ \lambda_0 $ on the boundary of the ball $ B_R = \{ \lambda \in \mathbb{C} \colon | \lambda | \leq R \} $, i.e. $ F( \lambda ) $ can be continued analytically to some larger domain. Hence we obtain a contradiction to the assumption that $ R < 1 $ is the radius of convergence of the Taylor series for $ F( \lambda ) $. Thus $ R = 1 $ and $ F( \lambda ) $ can be extended analytically to $ D_+ $. \qed

\end{document}